\documentclass{article}

\usepackage[utf8]{inputenc}
\usepackage{times}

\usepackage{soul}
\usepackage[small]{caption}

\usepackage{appendix}
\usepackage{geometry}
\usepackage{amsmath}
\usepackage{amsfonts}
\usepackage{amsthm}
\usepackage{amstext}
\usepackage{mathrsfs}
\usepackage{amssymb}
\usepackage{graphicx}
\usepackage{subfigure}
\usepackage{algorithm,algorithmicx}
\usepackage[noend]{algpseudocode}
\usepackage{bm}
\usepackage{color}
\usepackage{natbib}
\usepackage{hyperref}
\usepackage{url}
\usepackage{caption}
\usepackage{verbatim}

\usepackage{multicol}
\usepackage{multirow}

\makeatletter
\newcommand{\rmnum}[1]{\romannumeral #1}
\newcommand{\Rmnum}[1]{\expandafter\@slowromancap\romannumeral #1@}
\makeatother

\newtheorem{definition}{Definition}[section]
\newtheorem{lemma}{Lemma}[section]
\newtheorem{theorem}{Theorem}[section]

\title{Nash Equilibria of Two-round Auctions}
\author{Chulong Zhong$^1$, Xiang Yan$^2$, Yuyi Wang$^3$, Shuangping Huang$^1$, Jin Zhong$^1$\\
{$^1$South China University of Technology}\\
{$^2$Huawei Taylor Lab}\\
{$^3$CRRC $\lambda$ Lab}\\
zhongchulong7013@gmail.com
}
\date{}

\begin{document}

\maketitle

\begin{abstract}
In a two-round auction, a subset of bidders is selected (probabilistically), according to their bids in the first round, for the second round, where they can increase their bids. 
We formalize the two-round auction model, restricting the second round to a dominant strategy incentive compatible (DSIC) auction for the selected bidders. 
It turns out that, however, such two-round auctions are not directly DSIC, even if the probability of each bidder being selected for the second round is monotonic to its first bid, which is surprisingly counter-intuitive. 
We also illustrate the necessary and sufficient conditions of two-round auctions being DSIC. 
Besides, we characterize the Nash equilibria for untruthful two-round auctions. One can achieve better revenue performance by setting proper probability for selecting bidders for the second round compared with truthful two-round auctions.
\end{abstract}

\section{Introduction}
Auction design has long been an important problem in computational economics, as its wide applications in sponsored search, resource allocation, and blockchain.

In the history of auction theory, finding dominant strategy incentive compatible (DSIC) mechanisms is usually the main goal for auction designers.
This is because the DSIC condition assures that each bidder in the auction achieves its optimal utility through truthfully bidding, no matter what the other bidders do.
Moreover, Myerson's famous work \cite{ref_article3} has already provided an important characterization for the DSIC auctions of single items.

A recent observation on auctions happening in the industry is that there is usually more than one round in each auction. 
Precisely, bidders participating in the auction are asked to report their bids at the beginning. After eliminating some of the bidders by pre-specified rules, the remaining bidders are allowed to update their bids. 
Naturally, only increasing bids are allowed.
This leads to an interesting question: Whether a two-round auction is equivalent to a single-round auction in any sense?

One may conjecture the equivalence between the two-round auctions and the single-round auctions as long as the second rounds of the former one are exactly some DSIC auctions for a subset of bidders. 
However, this is not always true.
In this paper, we formalize the theoretical model for two-round auctions and derive the sufficient and necessary condition for a two-round auction being DSIC, provided its second round forms a single-round DSIC auction.

Furthermore, bidders in multi-round auctions are no longer in a single-parameter environment since the winning set or the final allocation of the single item depends on the multi-dimensional bids in each round. 
Thus, the revenue-maximization auction is not necessarily provided by Myerson's lemma. 
In this paper, we also try to characterize the Nash equilibria when a two-round auction is not DSIC.
We find that through applying tricks of multiple rounds, especially by properly setting the probability for selecting bidders to participate in the second-round auction, more revenue can be obtained under the corresponding unique Nash equilibrium, compared with the revenue-maximized single-round auction. 

\subsection{Related Work}
Efficiency improvements have been associated with a multi-round mechanism.
Perry and Reny~\cite{ref_article8,ref_book4} implement a two-round trick on the foundation of Vickrey's auction, compensating for the lack of efficiency of Vickrey's ex-post efficient multi-unit auction on the interdependent private valuation.
Rolfe~\cite{ref_article12} conducted an information-limited auction experiment for landholders in Australia, opposite the English auction.
The results revealed that landholders could gain information about the suitability of their proposals, whether their valuation is interdependent or not, with the aid of a multi-round mechanism and have an efficiency increment, particularly in the second round, with mild improvement between the second round and the third round.
However, there has not been any work formalizing the two-round auction model and providing algorithmic game theoretical analysis.

The multi-round mechanism could provoke greater revenue as well. 
Oral ascending auctions have a complete information giveaway~\cite{ref_article16}, like bids and bidder identities, to bidders, while multi-round auctions with interdependent values only have partial information, like only bids, told to the bidders which can avoid the collision caused by two blade effects of information sharing.
List~\cite{ref_article17} examines panel data and finds multi-round tricks will enhance the median bid for naive bidders, defined as offering only price information after each round, in a sealed-bid second price auction.
Katzman~\cite{ref_article9} raises an effective two-round sequential second price auction in 2-unit auctions generating equivalent expected revenue with Vickrey's auction, in which different items are sold in different rounds.
Also, some researchers summarize that multi-round auction promotes information sharing~\cite{ref_book5} and maximizes the total selling price in package auction~\cite{ref_article13}, in which multiple items segmented into fewer packages will be sold by the package. 
Cybernomics~\cite{ref_article14} found that simultaneous multi-round(SMR) auction used by FCC from 1995 creates higher revenue than CRA combinatorial auction in a multi-unit auction, sacrificing the efficiency to some extent.
Our observation for better revenue performance in untruthful two-round auctions than single-round auctions may be the first theoretical evidence.

\subsection{Organization} In Section \ref{model}, we introduce a model for two-round auctions, including necessary notations and preliminary. 
Section \ref{proof} proves the sufficient and necessary condition for a two-round auction being DSIC, summarized in Theorem \ref{lemma:two round dsic}. In Section \ref{ne condition}, we try to characterize the Nash equilibria if the auctions do not satisfy the DSIC condition. We focus on two simple types of auctions and find better revenue performance compared to the single-round model. 

\section{Two-round Auctions}
\label{model}

We consider a two-round auction with $n$ bidders $N=\{1,2,\dots,n\}$, competing for one single item. 
Each bidder $i$ has a private value for the item, denoted by $v_i$.
W.L.O.G., we assume $v_1\geq v_2\geq \dots \geq v_n$.
A two-round auction mechanism is denoted by $\mathcal{M}(\theta, \alpha, \boldsymbol{x}, \boldsymbol{y}, \boldsymbol{p})$, precisely:

\begin{itemize}
    \item In the first round, bidders bid $\boldsymbol{b'}=\{b'_1,\dots,b'_n\}$. 
    We re-order the index in $\boldsymbol{b'}$ and obtain $\boldsymbol{b}$ such that $b_1\geq b_2\dots\geq b_n$.
    Then $\alpha$ bidders are selected from those with the highest $\theta$ bids ($b_1,\dots,b_{\theta}$) to participate the second round, according to rule $\boldsymbol{y}=\{y_1,\dots,y_{\theta}\}$.
    Namely, the marginal probability of the bidder with bid $b_i$ being selected is $y_i$ (discussed later).
    \item In the second round, bidders are allowed to update their bids $\boldsymbol{s'}=\{s'_1,\dots,s'_{n}\}$ such that $s'_i\geq b'_i$ if bidder $i$ is selected to the second round. 
    (Otherwise, let $s'_i=0$ for completeness.) 
    Similarly, we re-order the index to obtain $\boldsymbol{s}$ such that $s_1\geq \dots\geq s_{\alpha}$.
    Then the winner of the item is determined by the allocation rule $\boldsymbol{x}(\boldsymbol{s})$ and its payment by $\boldsymbol{p}(\boldsymbol{s})$.
    Specifically, let $\beta$ be the number of bidders with a positive probability of allocation in the second round, i.e., $\beta=|\{i|x_i>0\}|$.
    \item Assume each bidder has quasi-linear utility, that is, its private value minus payment if it wins the item and zero otherwise. 
\end{itemize}

\begin{table*}
  \caption{Intuitive explanation of important mathematical notations}
  \label{tab:notation}
  \resizebox{\linewidth}{!}{
  \begin{tabular}{c|l}
  \hline
  $\theta$&the number of bidders who has a chance to be chosen into the second round\\
  \hline
  $\alpha$&the number of bidders who is actually  chosen into the second round\\
  \hline
  $\beta$&the number of bidders who has a chance to be chosen to be final winner in the second round\\
  \hline
  $y_i$&the probability of i-th highest bidder to enter the second round in the first round auction\\
  \hline
  $Y_i$&the probability of every $C_{\theta}^{\alpha}$ draw result, in which $\alpha$ winners are drawn from the top $\theta$ bidders\\
  \hline
  $x_j$&the probability of j-th highest bidder to be the final winner in the second round auction\\
  \hline
  $u(\cdot)$&the expected utility in the whole two-round auction\\
  \hline
  $u_2(\cdot)$&the utility in the second round auction\\
  \hline
\end{tabular}}
\end{table*}

\paragraph{Remark.}
Some important notations are explained in Table~\ref{tab:notation} intuitively.
We use the variables $\theta,\alpha,\beta$ and introduce the probability into the auction and get a more general model.
It's readily comprehensible that we let $\theta\ge\alpha\ge\beta\ge 1$.
Narrowly, if $\theta=\alpha,\beta=1$ the model degenerates into a common auction model without uncertainty totally determined by the bids. 
The model described in~\cite{ref_article10} corresponds to $\theta=\alpha=2,\beta=1$, and it is proved to have equivalent expected revenue with ascending open-bid auctions. 
On the other hand, our model can be applied to include reserve price $r_1$ for the first round and $r_2$ for the second round, where $r_1\le r_2$.

\subsection{Design Space for the First Round}

We consider the situation where the auctioneer decides the joint probability for selecting $\alpha$ bidders from the top $\theta$ bids in the first round.
This corresponds to $C(\theta,\alpha)=\frac{\theta!}{\alpha!(\theta-\alpha)!}$ possible combinations, denoted by $\mathbb{L}_\theta^\alpha(\cdot)$.
Precisely, for index set $\{1,2,\cdots,\theta\}$, $\mathbb{L}_\theta^\alpha(1,2,\cdots,\theta)$ ($\mathbb{L}_\theta^\alpha$ if there is not ambiguity) represents all its subsets of size $\alpha$.

In a two-round auction mechanism, the auctioneer should determine the probability $P(L^{(\alpha)})$ for each $L^{(\alpha)}\in \mathbb{L}_\theta^\alpha$.
Additionally, for each $L^{(\alpha)}\in \mathbb{L}_\theta^\alpha$, denote $\boldsymbol{s}_{L^{(\alpha)}}$ as the bidders' bids in the second round.
For simplicity, we focus on the marginal probability, that is, for bidder with the $i$-th highest bid, the probability it is selected to the second round is
$$y_i=\displaystyle{\frac{\sum_{L^{(\alpha-1)}\in\mathbb{L}_{\theta-1}^{\alpha-1}(\cdots i-1,i+1\cdots)}P(i,L^{(\alpha-1)})}{\sum_{L^{(\alpha)}\in\mathbb{L}_\theta^\alpha(\cdots i-1,i+1\cdots)}P(L^{(\alpha)})}},$$
and easily $\sum_{i=1}^\theta y_i=\alpha$.

On the other hand, on behalf of a single bidder, we re-use the notations $y(b,\boldsymbol{b}_{-i})$ to define the probability the bidder with bid $b$ in the first round is selected to the second round, for fixed other $n-1$ bids $\boldsymbol{b}_{-i}$.

\subsection{Preliminary}

For bidders participating in the second round, they are exactly participating in a single-round auction with a personalized reserve price, their bids in the first round.
Thus, we focus on those two-round auctions whose second rounds are DSIC and analyze whether the design for the first round matters.

\begin{definition}[DSIC]
An auction is dominant strategy incentive compatible (DSIC) if for each bidder, bidding truthfully optimizes its utility no matter what other bidders bid.
\end{definition}

\begin{definition}[Myerson's Lemma] A single round auction $(x(\boldsymbol{s'}),p(\boldsymbol{s'}))$ is DSIC if and only if, for any $s_\xi'$ and fixed $\boldsymbol{s'_{-\xi}},$\\
1. the allocation rule $x(s_\xi',\boldsymbol{s'_{-\xi}})$ is monotone non-decreasing on $s_\xi'$.\\
2. the payment rule accords with $p_\xi(\mathbf{s}^{(\alpha)},\mathbf{x}^{(\alpha)})=\int_{0}^{s_\xi'} z\cdot\frac{d}{dz}x(s_\xi',\boldsymbol{s'_{-\xi}}) dz$.
\end{definition}

Note that if the allocation rule takes the value of discrete constants, Myerson's lemma has its discrete form. 
It has monotone allocation rule $x_\beta\le\cdots\le x_3\le x_2\le x_1$ and payment $p_\xi(\mathbf{s}^{(\alpha)},\mathbf{x}^{(\alpha)})=\sum\limits_{j=\xi}^{\beta}s_{j+1}(x_j-x_{j+1})$.

\section{DSIC condition for two-round auctions}
\label{proof}

In this section, we prove the necessary and sufficient condition for any two-round auction being DSIC if its second round itself is DSIC.

\begin{theorem}
\label{lemma:two round dsic}
Consider a two-round auction $\mathcal{M}(\theta, \alpha, \boldsymbol{x}, \boldsymbol{y}, \boldsymbol{p})$.
Assume the allocation and payment rule of its second round $(\boldsymbol{x}, \boldsymbol{p})$ forms a DSIC auction for $\alpha$ bidders.
Then the two-round auction is also DSIC, if and only if:
\begin{enumerate}
    \item $y(b,\boldsymbol{b}_{-i})\le y(b_{\theta-\alpha+\beta},\boldsymbol{b}_{-i})$ for $b<b_{\theta-\alpha+\beta}$
    \item $y(b,\boldsymbol{b}_{-i})= y(b_{\theta-\alpha+\beta},\boldsymbol{b}_{-i})$ for $b\ge b_{\theta-\alpha+\beta}$,
\end{enumerate}
where $\beta=|\{i|x_i>0\}|$. For a discrete form, we have,
\begin{enumerate}
    \item $y_i\le y_{\theta-\alpha+\beta}$ for $i>\theta-\alpha+\beta$,
    \item $y_1=y_2=\cdots=y_{\theta-\alpha+\beta}$,
\end{enumerate}

\end{theorem}

Theorem \ref{lemma:two round dsic} assures that the probability to enter the second round remains equal before the $\theta-\alpha+\beta$ order and it has less probability when the bidders order after $\theta-\alpha+\beta$.
As shown in Fig~\ref{fig:two round dsic condition}, for $\theta-\alpha+\beta=3$, the DSIC condition is that $y_1=y_2=y_3>y_4,\sum{y_i}=3$, because $y_i$ means the chance to enter the second round for one person and there are 3 persons chosen into the the second round.
Note that in the first condition, it is not required that the probability to the second round must be monotone for bidders ranked after $\theta-\alpha+\beta$.
In fact, only bidders with the highest $\theta-\alpha+\beta$ in the first round have a chance to win in the second round, which is shown in the following lemma.
Therefore, the second condition prevents the $\theta-\alpha+\beta$ competitors' bids riskily to obtain more chance to enter the second round to raise their expected utility, which violates the DSIC definition.
Typically, we focus on the $\theta>\alpha\ge\beta$ because when $\theta=\alpha$, there is no uncertainty in the first round and we can conclude $y(b,\boldsymbol{b}_{-i})=1,b>b_{\theta}$, which absolutely satisfies Theorem~\ref{lemma:two round dsic}.
So when $\theta=\alpha$, the two-round auction is DSIC and equivalent to the single-round auction concerning social surplus and revenue.
\begin{figure}[ht]
  \centering
  \includegraphics[width=0.3\linewidth]{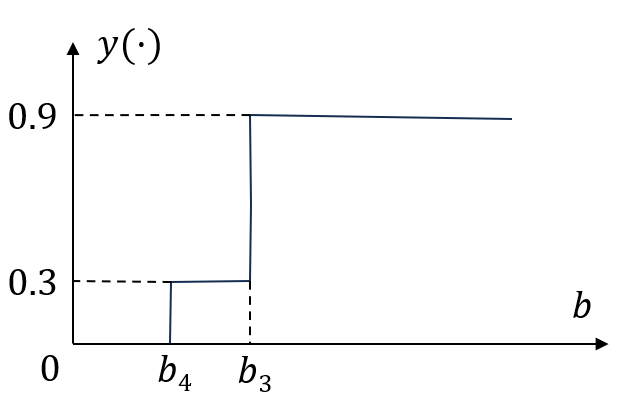}
  \caption{the example $y(\cdot)$ for DSIC condition of two-round auction when $\theta=4,\alpha=3,\beta=2$}
  \label{fig:two round dsic condition}
\end{figure}

\begin{lemma}[Line of competition]
    \label{competion line}
        $s_\beta\ge b_{\theta-\alpha+\beta}$
\end{lemma}
\begin{proof}
Let $L^{(\alpha)}_{min}\triangleq(\theta,\theta-1,\cdots,\theta-\alpha+1)\in\mathbb{L}^{\alpha}_{\theta}$, which corresponds to the case that bidders with the smallest bids selected to the second round.
Then $\forall$ $L^{(\alpha)}\in\mathbb{L}^{\alpha}_{\theta}$, $\forall$ $i\in L^{(\alpha)}_{min}$, $\exists$ $j\in L^{(\alpha)}$, $s.t.$ $b_{i}\le b_{j}$.

The bids defined by $L^{(\alpha)}_{min}$ can be regarded as \\
\centerline{$\{b_{\theta-\alpha+\alpha},b_{\theta-\alpha+(\alpha-1)},\cdots,b_{\theta-\alpha+1}\}$,}
while the bids in the second round are \\
\centerline{$\{s_{\alpha},s_{\alpha-1},\cdots,s_1\}$.}

Then we can get \\
\centerline{$s_1\ge b_{\theta-\alpha+1},s_2\ge b_{\theta-\alpha+2},\cdots,s_{\alpha}\ge b_{\theta-\alpha+\alpha}$.}
That is, for $\gamma\le\beta$, $s_\gamma\ge b_{\theta-\alpha+\gamma}$, and specifically, $s_\beta\ge b_{\theta-\alpha+\beta}$.
    
\end{proof}

Now we prove Theorem \ref{lemma:two round dsic}.
For bidder with private value $v$ and any other $n-1$ bidders' bids $\boldsymbol{b}_{-i}$ and $\boldsymbol{s}_{-j}$ fixed, let $u(\cdot)$ denote its utility in the whole two-round auction. 
Formally, when bidding truthfully, i.e. $b=s=v$, $u(v,\boldsymbol{b}_{-i},v,\boldsymbol{s}_{-j})=y(v,\boldsymbol{b}_{-i})\cdot u_2(v,\boldsymbol{s}_{-j})$, where $u_2(s,\boldsymbol{s}_{-j})$ is the utility function when the bidder participates the auction defined by the second round $(\boldsymbol{x}, \boldsymbol{p})$.
Consider two kinds of bidding methods:
\begin{enumerate}
\item Conservative: $b\le v$, but $s<v,s=v,s>v$ are all possible;
\item Risky: $v<b\le s$.
\end{enumerate}

We will prove conditions 1 and 2 are sufficient and necessary conditions for $u(v,\boldsymbol{b}_{-i},v,\boldsymbol{s}_{-j})\ge u(b,\boldsymbol{b}_{-i},s,\boldsymbol{s}_{-j})$.

\subsection{Sufficiency}

Since the mechanism $(\boldsymbol{x}, \boldsymbol{p})$ is DSIC, we know that for any $s$, $u_2(v,\boldsymbol{s}_{-j})\ge u_2(s,\boldsymbol{s}_{-j})$.
If $v>s_\beta\ge b_{\theta-\alpha+\beta}$,

under the condition of Theorem \ref{lemma:two round dsic}, for any $b$,
$$
y(v,\boldsymbol{b}_{-i})\ge y(b,\boldsymbol{b}_{-i}).
$$

The equality holds only when $b\ge b_{\theta-\alpha+\beta}$. 
Then $u(v,\boldsymbol{b}_{-i},v,\boldsymbol{s}_{-j})\ge u(b,b_{-i},s,s_{-j})$.

If $v\le s_\beta$, the bidder will not get the positive utility, i.e. $u_2(v,\boldsymbol{s}_{-j})\leq 0$.
Then for any $b$,
$$
u(b,\boldsymbol{b}_{-i},s,\boldsymbol{s}_{-j})=y(b,\boldsymbol{b}_{-i})\cdot u_2(s,\boldsymbol{s}_{-j})\le 0
$$
$$
u(v,\boldsymbol{b}_{-i},v,\boldsymbol{s}_{-j})=y(v,\boldsymbol{b}_{-i})\cdot u_2(v,\boldsymbol{s}_{-j})=0
$$
Thus, whether by conservative bidding $b\ge v$ or risky bidding $s\ge b>v$, $\quad u(v,\boldsymbol{b}_{-i},v,\boldsymbol{s}_{-j})\ge u(b,\boldsymbol{b}_{-i},s,\boldsymbol{s}_{-j})$.

\subsection{Necessity} Consider any probability function $y(\cdot)$ that does not satisfy condition in Theorem \ref{lemma:two round dsic}, we show there exist $v$, $\boldsymbol{b}_{-i}$, $\boldsymbol{s}_{-j}$, and $b$ or $s$, such that $u(b,\boldsymbol{b}_{-i},s,\boldsymbol{s}_{-j})>u(v,\boldsymbol{b}_{-i},v,\boldsymbol{s}_{-j})$.\\

\paragraph{Condition 2 is necessary.}
In conservative condition which is $b\le v$, if there exists $\epsilon_2 > \epsilon_1 \geq b_{\theta-\alpha+\beta}$, such that $y(\epsilon_1,\boldsymbol{b}_{-i}) > y(\epsilon_2,\boldsymbol{b}_{-i})$, (in discrete form, this means $y_i>y_j$ while $i>j$,) then we can let $b=\epsilon_1, s=v=\epsilon_2$, which means 
\begin{eqnarray}
u(b,\boldsymbol{b}_{-i},s,\boldsymbol{s}_{-j})&=&y(b,\boldsymbol{b}_{-i})\cdot u_2(s,\boldsymbol{s}_{-j})\notag\\
    &>&y(v,\boldsymbol{b}_{-i})\cdot u_2(v,\boldsymbol{s}_{-j})=u(v,\boldsymbol{b}_{-i},v,\boldsymbol{s}_{-j}). \notag
\end{eqnarray}

In risky condition which is $b>v$,if there exists $\epsilon_2 > \epsilon_1 \geq b_{\theta-\alpha+\beta}$, such that $y(\epsilon_1,\boldsymbol{b}_{-i}) < y(\epsilon_2,\boldsymbol{b}_{-i})$, (in discrete form, this means $y_i<y_j$ while $i>j$,) we can also construct $\boldsymbol{s}_{-j}$ such that $(\epsilon_1,\epsilon_2)\subset(s_\gamma,s_{\gamma-1})$ for some $\gamma \leq \beta$.
Then we can let $b=s=\epsilon_2,v=\epsilon_1$, which means $u_2(s,\boldsymbol{s}_{-j})=u_2(v,\boldsymbol{s}_{-j})$ (because otherwise, the second round is not DSIC), while $y(b,\boldsymbol{b}_{-i})>y(v,\boldsymbol{b}_{-i})$.
To sum up, $u(b,\boldsymbol{b}_{-i},s,\boldsymbol{s}_{-j}) > u(v,\boldsymbol{b}_{-i},v,\boldsymbol{s}_{-j})$ when $b>v$.

\paragraph{Under Condition 2, Condition 1 is also necessary.}
If there exists $\epsilon_3 < b_{\theta-\alpha+\beta}$, such that $y(\epsilon_3,\boldsymbol{b}_{-i}) > y(b_{\theta-\alpha+\beta},\boldsymbol{b}_{-i})$, (in discrete form, this means $y_i>y_{\theta-\alpha+\beta}$ while $i>\theta-\alpha+\beta$,) then we can let $b=\epsilon_3, s=v=b_{\theta-\alpha+\beta}$, which means
\begin{eqnarray}
u(b,\boldsymbol{b}_{-i},s,\boldsymbol{s}_{-j})&=&y(b,\boldsymbol{b}_{-i})\cdot u_2(s,\boldsymbol{s}_{-j})\notag\\
    &>&y(v,\boldsymbol{b}_{-i})\cdot u_2(v,\boldsymbol{s}_{-j})=u(v,\boldsymbol{b}_{-i},v,\boldsymbol{s}_{-j})\notag
\end{eqnarray}

This completes the proof of Theorem \ref{lemma:two round dsic}.

\section{Nash Equilibrium in Two-round Auctions}
\label{ne condition}


As the DSIC condition for two-round auctions requires that bidders with top bids have equal probabilities to the second round, it is interesting to consider those two-round auctions without such a condition.
Specifically, we focus on the popular monotonic condition, that the probability each bidder is selected for the second round is monotonicly increasing w.r.t.\ its bid, given others' bids fixed.
Then we try to characterize the Nash equilibrium (NE) bidding strategies in two-round auctions.

Under the monotonic condition, it is direct that the bidder must bid the same in both rounds, and we only need to consider bidders with the highest $\theta-\alpha+\beta$ values.
Thus, for a fixed mechanism and bidders' bid $\boldsymbol{b}'$, the utility of each bidder $i$ can be represented by the order of its bid in the first round, denoted by $a_i(\boldsymbol{b}')$.
That is, $$u(v_i,a_i(\boldsymbol{b}')) = \mathbb{E}_{L^{(\alpha)}}[u_2(v_i,\boldsymbol{x}(\boldsymbol{s}_{L^{(\alpha)}}),\boldsymbol{p}(\boldsymbol{s}_{L^{(\alpha)}}))],$$ 
where $$u_2(v_i,\boldsymbol{x}(\boldsymbol{s}_{L^{(\alpha)}}),\boldsymbol{p}(\boldsymbol{s}_{L^{(\alpha)}}))=v_i\cdot x_j-p_j$$ 
if $a_i\in L^{(\alpha)}$ and $b_i$ is the $j$-th highest in $\boldsymbol{s}_{L^{(\alpha)}}$, $\boldsymbol{p}$ is given by Myerson's lemma.
Now we define NE following the one given by~\cite{ref_article4}.
Recalled that $u(\cdot)$ was defined as utility in the whole two-round auction in this papers, which is different from the article~\cite{ref_article4} focusing on the single-round auction.
Intuitively, under NE, no bidder can benefit by deviating to another position. 

\begin{definition}
\label{definition:ne}
    A bid profile $\boldsymbol{b}'=(b'_i,\boldsymbol{b}'_{-i})$ is a Nash equilibrium (NE) if it satisfies

    $$u(v_i,a_i(\boldsymbol{b}'))\geq u(v_i,t),\forall t=1,2,\cdots,(\theta-\alpha+\beta).$$
    Specifically, we use the notation $\mathbf{Rank}[a_1,a_2,\dots,a_n]$ to represent the bidding positions for bidders in the first round, then we say $\mathbf{Rank}[a_1,a_2,\dots,a_n]$ is NE if the corresponding utilities of bidders satisfy the condition.

\end{definition}

\paragraph{Auction Mode}

We set $\theta-\alpha=\Delta n$, and use $\mathcal{N}_{\Delta n,\beta}$ to denote the mode of an auction.
As representing NE for all modes is too complicated, we characterize NE for the simple mode $\mathcal{N}_{1,1}$ and present a simple analysis for $\mathcal{N}_{1,2}$.

\subsection{Nash equilibria of mode $\mathcal{N}_{1,1}$}
In mode $\mathcal{N}_{1,1}$, that is $\alpha=\theta-1,\beta=1$, the second round degenerates into the second-price auction.
Now $\theta-\alpha+\beta=\Delta n+\beta=2$, it means only bidders $1$ and $2$ have the probability to win in the second round and risky bidding will only happen between these $(\theta-\alpha+\beta)=2$ bidders.
In the meantime, the others will bid truthfully in Nash equilibrium which is $s_j'=b'_j=v_j,3\ge j\ge \theta$.
In this section, we discuss the conditions and the revenue of Nash equilibria of mode $\mathcal{N}_{1,1}$.

We have defined $Y_i$ as the probability of each combination when choosing $\alpha$ bidders from $\theta(\theta\ge 3)$ ones.
\begin{equation*}
y_1>y_2>y_j,j=3,\cdots,\theta
\end{equation*}
\begin{equation*}
\displaystyle{Y_{\theta-i+1}=\lambda\cdot\prod\limits_{j=1}^{\theta}{y_j}\cdot\frac{1-y_i}{y_i},\quad\sum\limits_{i=1}^{\theta}{Y_i}=1}
\end{equation*}
\begin{equation*}
Y_i>Y_{\theta-1}>Y_{\theta},i=1,\cdots,\theta-2
\end{equation*}
To satisfy $\sum\limits_{i=1}^{\theta}{Y_i}=1$, we get
\begin{equation}
\displaystyle{Y_{\theta-1}<\frac{1}{\theta-1},Y_{\theta}<\frac{1}{\theta}}\label{eq1}
\end{equation}
Otherwise we get $\sum\limits_{i=1}^{\theta}{Y_i}>1.$

\subsubsection{Calculation.}From the Definition \ref{definition:ne}, we find the sufficient and necessary condition of the Nash equilibria of mode $\mathcal{N}_{1,1}$ can be expressed as 2 inequalities for each ranking list,
$$ \left\{
\begin{aligned}
u(v_{i_1},1)>u(v_{i_1},2),a_{i_1}=1, \\
u(v_{i_2},1)<u(v_{i_2},2),a_{i_2}=2, 
\end{aligned}
\right.
$$
$$
for\quad i_1,i_2\in\{1,2\}.
$$
(\rmnum{1})For $\mathbf{Rank}[1,2, 3,\cdots,\theta],a_1=1,a_2=2$,
\begin{equation*}
    u(v_2,2)=Y_{\theta}(v_2-v_3),u(v_1,1)=\displaystyle{(\sum\limits_{i=1}^{\theta-2}{Y_i})(v_1-b_2')+Y_{\theta-1}(v_1-v_3)}
\end{equation*}
For $u(v_2,2)<u(v_2,1),u(v_1,2)>u(v_1,1)$, the $\mathbf{Rank}[2,1,3,\cdots, \theta]$ is a Nash equilibrium of risky bidding, which satisfies
\begin{equation}
\displaystyle{\frac{v_1-v_2}{v_2-v_3}<\frac{Y_{\theta-1}-Y_{\theta}}{\sum\limits_{i=1}^{\theta-2}Y_i}<\frac{b_2'-v_1}{v_1-v_3}}\label{con1}
\end{equation}

(\rmnum{2})For $\mathbf{Rank}[2,1, 3,\cdots,\theta],a_1=2,a_2=1$,
\begin{equation*}
    u(v_1,2)=Y_{\theta}(v_1-v_3),u(v_2,1)=\displaystyle{(\sum\limits_{i=1}^{\theta-2}{Y_i})(v_2-v_1)+Y_{\theta-1}(v_2-v_3)}
\end{equation*}
For $u(v_2,2)>u(v_2,1)$, the $\mathbf{Rank}[1,2,3,\cdots, \theta]$ is a Nash equilibrium of truthful bidding, which satisfies
\begin{equation}
\displaystyle{\frac{v_1-v_2}{v_2-v_3}>\frac{Y_{\theta-1}-Y_{\theta}}{\sum\limits_{i=1}^{\theta-2}Y_i}}\label{con2}
\end{equation}
If $u(v_2,2)>u(v_2,1)$, the $v_2$ bidder must bid truthfully so $b_2'=v_2$, and it must satisfy $u(v_1,1)>u(v_1,2)$ for the nonnegative item $(\sum\limits_{i=1}^{\theta-2}{Y_i})(v_1-b_2')$.\\

\subsubsection{Analysis.}For inequality (\ref{eq1}),\\
\begin{center}
$\displaystyle{\frac{\partial\Bigg(\frac{Y_{\theta-1}-Y_{\theta}}{1-Y_{\theta-1}-Y_{\theta}}\Bigg)}{\partial Y_{\theta}}<0},\displaystyle{\frac{\partial\Bigg(\frac{Y_{\theta-1}-Y_{\theta}}{1-Y_{\theta-1}-Y_{\theta}}\Bigg)}{\partial Y_{\theta-1}}>0}\Rightarrow \displaystyle{\frac{Y_{\theta-1}-Y_{\theta}}{\sum\limits_{i=1}^{\theta-2}Y_i}}<\frac{1}{\theta-2}$
\end{center}
If $\displaystyle{\frac{v_1-v_2}{v_2-v_3}>\frac{1}{\theta-2}}$, the risky bidding doesn't exist anymore, which means there is only a single Nash equilibrium. 
On the condition of $\displaystyle{\frac{v_1-v_2}{v_2-v_3}<\frac{1}{\theta-2}}$, there may exist the risky bidding.
\subsection{Revenue of mode $\mathcal{N}_{1,1}$}
We set the ranking list as the bidders' number ranking by the order in the final round.
For example, $\mathbf{Rank}[2,1,3,\cdots, \theta]$ means $s_2'>s_1'>s_3'>\cdots>s_{\theta}'$, which is a typical Nash equilibrium when $s_j'=b_j'=v_j,2\ge j\ge \theta,s_1'=b_1'=v_1,s_2'=b_2'>v_1$.
We define the revenue of the specified Nash equilibrium as
\begin{equation*}
\mathbf{REV}(\mathbf{Rank}[a_1,a_2,\dots,a_\theta])=\sum\limits_{i=1}^{C(\theta,\alpha)}Y_i\cdot p_i
\end{equation*}
\subsubsection{Calculation.}For the 2 Nash equilibrium in $\mathcal{N}_{1,1}$,
\begin{equation*}
\mathbf{REV}(\mathbf{Rank}[2,1,3,\dots,\theta])=(\sum\limits_{i=1}^{\theta-2}Y_i)v_1+(Y_{\theta-1}+Y_{\theta})v_3
\end{equation*}
\begin{equation*}
\mathbf{REV}(\mathbf{Rank}[1,2,3,\cdots,\theta])=(\sum\limits_{i=1}^{\theta-2}Y_i)v_2+(Y_{\theta-1}+Y_{\theta})v_3
\end{equation*}
It is not right to conclude that $\mathbf{REV}(\mathbf{Rank}[2,1,3,\cdots,\theta])$is greater than $\mathbf{REV}(\mathbf{Rank}[1,2,3,\cdots,\theta])$ since $\mathbf{Y}$ in $(\mathbf{Rank}[2,1,3,\cdots,\theta])$ and $(\mathbf{Rank}[1,2,3,\cdots,\theta])$ satisfy condition $(\ref{con1})$ and condition $(\ref{con2})$ seperately.
To prevent confusion, we let $\mathbf{Y}$ be probability setting in truthful bidding and $\mathbf{Y'}$ in risky bidding in the following section.
Because of inequality (\ref{eq1}),
\begin{equation*}
    \displaystyle{\mathbf{REV}(\mathbf{Rank}[2,1,3,\cdots,\theta])\in(\frac{\theta^2-3\theta+1}{\theta^2-\theta}v_1+\frac{2\theta-1}{\theta^2-\theta}v_3,v_1)}
\end{equation*}
\begin{equation*}
    \displaystyle{\mathbf{REV}(\mathbf{Rank}[1,2,3,\cdots,\theta])\in(\frac{\theta^2-3\theta+1}{\theta^2-\theta}v_2+\frac{2\theta-1}{\theta^2-\theta}v_3,v_2)}
\end{equation*}
we need to know,
\begin{equation*}
    \displaystyle{sup\{\mathbf{REV}(\mathbf{Rank}[1,2,3,\cdots,\theta])\}<sup\{\mathbf{REV}(\mathbf{Rank}[2,1,3,\cdots,\theta])\}}
\end{equation*}
\begin{equation*}
    \displaystyle{inf\{\mathbf{REV}(\mathbf{Rank}[1,2,3,\cdots,\theta])\}<inf\{\mathbf{REV}(\mathbf{Rank}[2,1,3,\cdots,\theta])\}}
\end{equation*}
\subsubsection{Analysis.}Recalled that $\theta$ is the number of bidders who have the probability to be chosen into the second round.
If we increase the $\theta$, absolutely we can increase the infimum of revenue because the coefficient of $v_3$ will decrease by inequality (\ref{eq1}).
However, we also decrease the probability of the existence of risky bidding by condition (\ref{con1}).
There are 3 meaningful Lemmas about the revenue of risky bidding in the following section.

Notice $v_1$ is not a tight supremum for $\mathbf{REV}(\mathbf{Rank}[2,1,3,\cdots,\theta])$.
Because $\sum\limits_{i=1}^{\theta-2}Y_i'$ must tend to $1$ if revenue of risky bidding tends to $v_1$.
At the same time, the more $\sum\limits_{i=1}^{\theta-2}Y_i'$ tends to $1$, the more difficulty it satisfies condition $(\ref{con1})$ for specified gap between $v_1$ and $v_2$.
\begin{center}
    $\lim\limits_{v_2\rightarrow v_1}\mathbf{REV}(\mathbf{Rank}[2,1,3,\cdots,\theta])=v_1$
\end{center}
Only when $v_2$ tends to $v_1$, the revenue of risky bidding can approach $v_1$ but not violate condition $(\ref{con1})$ as the above limits expression shows. So we find tighter supremum for the revenue of risky bidding in Lemma $\ref{lemma:n_11 supremum}$ through taking condition $(\ref{con1})$ into consideration as well.

\begin{lemma}
    \label{lemma:n_11 supremum}
        the revenue of the risky bidding in the mode $\mathcal{N}_{1,1}$ has a supremum and the supremum approaches the $v_1$ when the gap between $v_2$ and $v_1$ decreases.
    \end{lemma}
    \begin{proof}
    From condition $(\ref{con1})$,we get
    \begin{center}
    $v_3<v_2-\displaystyle{\frac{\sum\limits_{i=1}^{\theta-2}Y'_i}{Y'_{\theta-1}-Y'_{\theta}}(v_1-v_2)}$
    \end{center}
    Replacing $v_3$ we get that,
    \begin{eqnarray}
    &&\mathbf{REV}(\mathbf{Rank}[2,1,3,\cdots,\theta])\notag\\
    &<&(\sum\limits_{i=1}^{\theta-2}Y'_i)v_1+(Y'_{\theta-1}+Y'_{\theta})(v_2-\displaystyle{\frac{\sum\limits_{i=1}^{\theta-2}Y'_i(v_1-v_2)}{Y'_{\theta-1}-Y'_{\theta}}})\notag\\
    &=&\displaystyle{\frac{(Y'_{\theta-1}-Y'_{\theta})\sum\limits_{i=1}^{\theta-2}Y'_i}{Y'_{\theta-1}-Y'_{\theta}}v_1-\frac{(Y'_{\theta-1}+Y'_{\theta})\sum\limits_{i=1}^{\theta-2}Y'_i}{Y'_{\theta-1}-Y'_{\theta}}v_1} \displaystyle{+\frac{(Y'_{\theta-1}+Y'_{\theta})(\sum\limits_{i=1}^{\theta-1}Y'_i-Y'_{\theta})}{Y'_{\theta-1}-Y'_{\theta}}v_2}\notag\\
    &=&\displaystyle{\frac{(Y'_{\theta-1}+Y'_{\theta})(1-2Y'_{\theta})}{Y'_{\theta-1}-Y'_{\theta}}v_2-\frac{2Y'_{\theta}\sum\limits_{i=1}^{\theta-2}Y'_i}{Y'_{\theta-1}-Y'_{\theta}}v_1}\notag,
\end{eqnarray}
\begin{eqnarray}
    \lim\limits_{v_2\rightarrow v_1}sup\{\mathbf{REV}(\mathbf{Rank}[2,1,3,\cdots,\theta])\}
    &=&\lim\limits_{v_2\rightarrow v_1}\displaystyle{\frac{(Y'_{\theta-1}+Y'_{\theta})v_2-2Y'_{\theta}(\sum\limits_{i=1}^{\theta-2}Y_iv_1+(Y'_{\theta-1}+Y'_{\theta})v_2)}{Y'_{\theta-1}-Y'_{\theta}}}\notag\\
    &=&\displaystyle{\frac{(Y'_{\theta-1}+Y'_{\theta})v_1-2Y'_{\theta}\cdot 1\cdot v_1}{Y'_{\theta-1}-Y'_{\theta}}}\notag\\
    &=&v_1.\notag
    \end{eqnarray}
    
    \end{proof} 
Since the revenue of the single-round second price auction is $v_2$, any more complicated auction whose revenue doesn't surpass $v_2$ is meaningless.
For example $\mathbf{REV}(\mathbf{Rank}[1,2,3,\cdots,\theta])<v_2$ is meaningless in mode $\mathcal{N}_{1,1}$. 
It is significant for the auction designers to know the necessary setting to keep the meaningfulness in the two-round auction. 
What Lemma \ref{lemma:n_11 infimum} explains is that for proper $\mathbf{Y}$ in a two-round auction there is Nash equilibrium in risky bidding and the auctioneer will profit more in the single-round auction.
\begin{lemma}
    \label{lemma:n_11 infimum}
        Only when $Y'_{\theta}+Y'_{\theta-1}<\displaystyle{\frac{v_1-v_2}{v_1-v_3}}$, the two-round risky bidding has better revenue than truthful single-round auction.
    \end{lemma}
    \begin{proof}
    \begin{eqnarray}
    \mathbf{REV}(\mathbf{Rank}[2,1,3,\cdots,\theta])&>&v_2\notag,\\
    (1-Y'_{\theta-1}-Y'_{\theta})v_1+(Y'_{\theta-1}+Y'_{\theta})v_3&>&v_2\notag,\\
    (1-Y'_{\theta-1}-Y'_{\theta})v_1+(Y'_{\theta-1}+Y'_{\theta}-1)v_3&>&v_2-v_3\notag,\\
    (1-Y'_{\theta-1}-Y'_{\theta})(v_1-v_3)&>&v_2-v_3\notag,\\
    (1-Y'_{\theta-1}-Y'_{\theta})&>&\displaystyle{\frac{v_2-v_3}{v_1-v_3}}\notag,\\
    Y'_{\theta}+Y'_{\theta-1}&<&\displaystyle{\frac{v_1-v_2}{v_1-v_3}}.\notag
    \end{eqnarray}
    
    \end{proof}

However, Lemma \ref{lemma:n_11 infimum} conflicts with condition (\ref{con1}), which means the risky bidding and better revenue cannot be fulfilled simultaneously. In other words, there is no specific $\mathbf{Y}$ setting for a risky bidding that has a better revenue compared with the truthful single-round auction.

What Lemma \ref{lemma:n_11 all} explains is that in a two-round auction, the revenue of risky bidding is better than that of truthful bidding.
In other words, it's important for the auctioneer to know how to set proper $\mathbf{Y'}$ to reach risky Nash equilibrium to profit more.
Lemma \ref{lemma:n_11 all} provides a way to get proper $\mathbf{Y'}$ in the risky bidding based on the known $\mathbf{Y}$ in the truthful bidding.

\begin{lemma}
    \label{lemma:n_11 all}
        No matter what allocation rule is set in the truthful two-round bidding, we can find a corresponding one in the risky two-round bidding which have a better revenue performance.
    \end{lemma}
    \begin{proof}
    For every $\mathbf{Y}=(Y_1,Y_2,\cdots,Y_{\theta})$ satisfying condition (\ref{con2}), there exists $\mathbf{Y'}=(Y_1',Y_2',\cdots,Y_{\theta}')$ satisfying condition (\ref{con1}),subjecting to $\mathbf{REV}(\mathbf{Rank}[1,2,3,\cdots,\theta])<\mathbf{REV}(\mathbf{Rank}[2,1,3,\cdots,\theta])$\\
    
    $(\rmnum{1})$For the circumstance of $\displaystyle{\frac{v_1-v_2}{v_2-v_3}<\frac{Y_{\theta-1}+Y_{\theta}}{1-Y_{\theta-1}-Y_{\theta}}}$, we implement fractional deformation and get,
    $$ 
    Y_{\theta}+Y_{\theta-1}>\displaystyle{\frac{v_1-v_2}{v_1-v_3}}.
    $$
    When $\sum\limits_{i=1}^{\theta-2}Y_i=\sum\limits_{i=1}^{\theta-2}Y_i',Y_{\theta-1}+Y_{\theta}=Y_{\theta-1}'+Y_{\theta}'$, we can deduce,
    \begin{eqnarray}
       \frac{Y_{\theta-1}'-Y_{\theta}'}{1-Y_{\theta-1}'-Y_{\theta}'}&<&\lim\limits_{Y_{\theta-1}'\rightarrow Y_{\theta-1}+Y_{\theta},Y_{\theta}'\rightarrow 0}\Bigg(\frac{Y_{\theta-1}'-Y_{\theta}'}{1-Y_{\theta-1}'-Y_{\theta}'}\Bigg)  
       \notag\\ 
       &=&\frac{Y_{\theta-1}+Y_{\theta}}{1-Y_{\theta-1}-Y_{\theta}}, \notag 
    \end{eqnarray}
        
    so there exists $ (Y_{\theta-1}'-Y_{\theta}')>\displaystyle{\frac{v_1-v_2}{v_2-v_3}(1-Y_{\theta-1}-Y_{\theta})}$ subjecting to $\mathbf{Y}$ satisfying condition $(\ref{con2})$.\\
    $(\rmnum{2})$ For the circumstance of 
    $\displaystyle{\frac{v_1-v_2}{v_2-v_3}>\frac{Y_{\theta-1}+Y_{\theta}}{1-Y_{\theta-1}-Y_{\theta}}}$, we implement fractional deformation and get,\\
    $$
    Y_{\theta}+Y_{\theta-1}<\displaystyle{\frac{v_1-v_2}{v_1-v_3}}.
    $$
    Only when$\sum\limits_{i=1}^{\theta-2}Y_i>\sum\limits_{i=1}^{\theta-2}Y_i',Y_{\theta-1}+Y_{\theta}<Y_{\theta-1}'+Y_{\theta}'$, $\mathbf{Y'}$ has the probability to satisfy condition $(\ref{con2})$, we set $Y_{\theta-1}+Y_{\theta}+\Delta Y=Y_{\theta-1}'+Y_{\theta}'$.\\
    To satisfy $$(1-Y_{\theta-1}'-Y_{\theta}')v_1+(Y_{\theta-1}'+Y_{\theta}')v_3>(1-Y_{\theta-1}-Y_{\theta})v_2+(Y_{\theta-1}+Y_{\theta})v_3,$$\\
    we finally get that $\displaystyle{\frac{v_1-v_2}{v_1-v_3}>\frac{\Delta Y}{1-Y_{\theta-1}-Y_{\theta}}}.$
    
    \end{proof}
\subsubsection{Experiments.} When designing a proper two-round auction, we need to precisely construct probability $\boldsymbol{Y}$ of each draw result for higher revenue.
For example, for $\theta=4,\alpha=3,C_4^3=4$, the $\boldsymbol{Y}=(Y_1,Y_2,Y_3,Y_4),\sum Y_i=1$ need to be set.
To compare the revenue of risky bidding and truthful bidding, we modulated a couples of control experiments by setting $\boldsymbol{Y}$ randomly on the computer for 10,000 times.
$\boldsymbol{Y}$ follows uniform distribution.
The different setting of $\boldsymbol{Y}$ satisfies different NE, so we add up how many times they take place separately and calculate the average revenue $u(\cdot)$ divided by their times.

In Table~\ref{tab:v1v2 gap}, when the gap between $v_1$ and $v_2$ decreases, the average revenue increases in both NE and it proves the correctness of increasing supremum in Lemma~\ref{lemma:n_11 supremum}. 
In Table~\ref{tab:theta}, with the $\theta$ growing, the $Y_{\theta-1}+Y_{\theta}$ declines.
As a result, it satisfies Lemma~\ref{lemma:n_11 infimum} more easily, and the average revenue increases.
The more $\sum\limits_{i=1}^{\theta-2}Y_i'$ tends to $1$, the more difficulty it satisfies condition $(\ref{con1})$ so the risky NE is hard to take place with the growth of $\theta$. 
Specially in experiment 8, because $\displaystyle{\frac{v_1-v_2}{v_2-v_3}=0.25=\frac{1}{\theta-2}}=\frac{1}{6-2}$, the risky bidding doesn't exist anymore.

The average revenue of risky bidding is usually greater than that of the truthful bidding in both Table~\ref{tab:v1v2 gap},\ref{tab:theta}, and it accords to Lemma~\ref{lemma:n_11 all}.
Especially in experiment 4, Lemma~\ref{lemma:n_11 infimum} is hard to satisfy so the profitable risky NE is of a small portion.
As a result, the revenue is slightly less than that in truthful NE and it also emphasizes the significance of the setting of $\boldsymbol{Y}$ by the formula given by Lemma above.

Finally, we set $\mathbf{v}=(v_1,v_2,v_3)$ randomly 1000 times, which $\mathbf{v}\sim U(0,1000)$. 
For each valuation $\mathbf{v}$, we set $\mathbf{Y}$ randomly 10,000 times and record only the maximum revenue in different NE.
Then we calculate the average revenue across 1000 different valuations as shown in Table~\ref{tab:random vy}.
Absolutely, we abandon some valuations when either risky bidding or truthful bidding doesn't exist.
Easily we can discover that with the growth of $\theta$, the revenue grows but the incremental ratio of risky bidding declines compared to the truthful bidding.

\begin{table}[ht!]
\begin{center}
    \caption{Control experiment on the gap between $v_1$ and $v_2$}
    \label{tab:v1v2 gap}
    \begin{tabular}{|c|c|c|c|c|c|}
      \hline
      \multirow{2}{*}{}&\multirow{2}{*}{$(v_1,v_2,v_3)$}&\multirow{2}{*}{$\theta$}&\multicolumn{2}{|c|}{average revenue across $\mathbf{Y}$}&\multirow{2}{*}{increment}\\
      \cline{4-5}
      \multirow{2}{*}{}&\multirow{2}{*}{}&\multirow{2}{*}{}&risky bidding&truthful bidding&\multirow{2}{*}{}\\
      \hline
      1&(450,350,200)&3&325.49&279.04&16.65\%\\
      \hline
      2&(450,400,200)&3&327.07&307.26&6.45\%\\
      \hline
      3&(450,425,200)&3&329.46&322.23&2.24\%\\
      \hline
      4&(450,440,200)&3&330.10&332.57&-0.74\%\\
      \hline
    \end{tabular}
\end{center}
\end{table}

\begin{table}[ht!]
\begin{center}
    \caption{Control experiment on $\theta$}
    \label{tab:theta}
    \begin{tabular}{|c|c|c|c|c|c|c|}
      \hline
      \multirow{2}{*}{}&\multirow{2}{*}{$(v_1,v_2,v_3)$}&\multirow{2}{*}{$\theta$}&\multicolumn{2}{|c|}{average revenue across $\mathbf{Y}$}&\multirow{2}{*}{increment}\\
      \cline{4-5}
      \multirow{2}{*}{}&\multirow{2}{*}{}&\multirow{2}{*}{}&risky bidding&truthful bidding&\multirow{2}{*}{}\\
      \hline
      5&(450,400,200)&3&328.20&307.57&6.71\%\\
      \hline
      6&(450,400,200)&4&373.38&343.82&8.60\%\\
      \hline
      7&(450,400,200)&5&392.21&362.21&8.28\%\\
      \hline
      8&(450,400,200)&6&/&372.98&/\\
      \hline
    \end{tabular}
\end{center}
\end{table}

\begin{table}[ht!]
\begin{center}
    \caption{Final experiment in random valuation}
    \label{tab:random vy}
    \begin{tabular}{|c|c|c|c|c|c|}
      \hline
      \multirow{2}{*}{}&\multirow{2}{*}{$\theta$}&\multicolumn{2}{|c|}{average revenue across $\mathbf{v}$}&\multirow{2}{*}{increment}\\
      \cline{3-4}
      \multirow{2}{*}{}&\multirow{2}{*}{}&risky bidding&truthful bidding&\multirow{2}{*}{}\\
      \hline
      9&3&382.81&367.18&$4.26\%$\\
      \hline
      10&4&398.43&395.01&$0.87\%$\\
      \hline
      11&5&409.87&407.83&$0.50\%$\\
      \hline
      12&6&412.66&412.18&$0.12\%$\\
      \hline
    \end{tabular}
\end{center}
\end{table}

\subsection{Nash equilibria of mode $\mathcal{N}_{1,2}$}
\label{n12}
Because of the complexity of the total depictions of equilibria of mode $\mathcal{N}_{1,2}$, we only discuss a specific example for $\theta=4,\alpha=3,\beta=2$.
Also for the symmetry of the result, we only discuss the sufficient condition of the Nash equilibria.
Since $\theta-\alpha+\beta=3$, we need $y_1=y_2=y_3$ to guarantee the auction strategy is DSIC.
If we set $y_1>y_2>y_3>y_4$ on the popular monotonic condition, there exists $(\theta-\alpha+\beta)!=6$ kinds of Nash equilibrium for a different arrangement of orders that someone may bid riskily surpassing their valuation.
The conditions of 6 kinds of Nash equilibrium are as follows, furthermore, not all the Nash equilibrium exist for specific $\mathbf{v}$.

\subsubsection{Calculation.}The sufficient condition of the Nash equilibria satisfying the Definition \ref{definition:ne} can be expressed by six inequalities for each kind of ranking list,
$$ \left\{
\begin{aligned}
u(v_{i_1},1)>u(v_{i_1},2),u(v_{i_1},2)>u(v_{i_1},3),a_{i_1}=1, \\
u(v_{i_2},1)<u(v_{i_2},2),u(v_{i_2},2)>u(v_{i_2},3),a_{i_2}=2, \\
u(v_{i_3},1)<u(v_{i_3},2),u(v_{i_3},2)<u(v_{i_3},3),a_{i_3}=3,
\end{aligned}
\right.
$$
$$
for\quad i_1,i_2,i_3\in\{1,2,3\}.
$$

And the results are given in Table~\ref{tab:table1},\ref{tab:table2},\ref{tab:table3} for,
\begin{center}
    $V_{2,3}(d_1;d_3,d_4)=\displaystyle{\frac{d_3-d_1}{d_1-d_4},\frac{\partial V_{2,3}}{\partial d_3}>0,\frac{\partial V_{2,3}}{\partial d_1}<0},d_1>d_4$,\\
    $V_{1,2}(d_1;d_2;d_3,d_4)=\displaystyle{\frac{(d_2-d_1)(x_1-x_2)}{(d_3-d_4)x_2-(d_3-d_1)x_1},\frac{\partial V_{1,2}}{\partial d_2}>0,\frac{\partial V_{1,2}}{\partial d_3}>0,\frac{\partial V_{1,2}}{\partial d_1}<0}$, $d_1<d_2$,\\
    $Y_{2,3}=\displaystyle{\frac{(Y_2-Y_3)x_2}{Y_1x_2+Y_4(x_1-x_2)}}$,\\
    $Y_{1,2}=\displaystyle{\frac{Y_3-Y_4}{Y_1-Y_2}}$.\\
\end{center}

\begin{table}[ht!]
\begin{center}
    \caption{Condition of Nash equilibrium: Part 1}
    \label{tab:table1}
    \resizebox{0.8\linewidth}{!}{
    \begin{tabular}{|c|c|c|c|}
      \hline
      \multicolumn{2}{|c|}{inverse number}&0&1\\
      \hline
      \multicolumn{2}{|c|}{Nash equilibrium}&$v_1(s_1')>v_2(s_2')>v_3(s_3')$&$v_1(s_1')>s_3'>v_2(s_2')$\\
      \hline
      \multirow{2}{*}{bidder of $v_1$}&$2^{nd}\leftrightarrow3^{rd}$& $\underline{V_{2,3}(v_1;v_3,v_4)<Y_{2,3}}$&$\underline{V_{2,3}(v_1;v_2,v_4)<Y_{2,3}}$\\
      \cline{2-4}
      \multirow{2}{*}{}&$1^{st}\leftrightarrow2^{nd}$&$\underline{V_{1,2}(v_1;v_2;v_3,v_4)<Y_{1,2}}$&$V_{1,2}(v_1;b_3';v_2,v_4)<Y_{1,2}$\\
      \hline
      \multirow{2}{*}{bidder of $v_2$}&$2^{nd}\leftrightarrow3^{rd}$&$\underline{V_{2,3}(v_2;v_3,v_4)<Y_{2,3}}$&$V_{2,3}(v_2;b_3',v_4)>Y_{2,3}$\\
      \cline{2-4}
      \multirow{2}{*}{}&$1^{st}\leftrightarrow2^{nd}$&$V_{1,2}(v_2;v_1;v_3,v_4)>Y_{1,2}$&$V_{1,2}(v_2;v_1;b_3',v_4)>Y_{1,2}$\\
      \hline
      \multirow{2}{*}{bidder of $v_3$}&$2^{nd}\leftrightarrow3^{rd}$&$V_{2,3}(v_3;v_2,v_4)>Y_{2,3}$&$V_{2,3}(v_3;v_2,v_4)<Y_{2,3}$\\
      \cline{2-4}
      \multirow{2}{*}{}&$1^{st}\leftrightarrow2^{nd}$&$V_{1,2}(v_3;v_1;v_2,v_4)>Y_{1,2}$&$V_{1,2}(v_3;v_1;v_2,v_4)>Y_{1,2}$\\
      \hline
    \end{tabular}}
\end{center}
\end{table}

\begin{table}[ht!]
\begin{center}
    \caption{Condition of Nash equilibrium: Part 2}
    \label{tab:table2}
    \resizebox{0.8\linewidth}{!}{
    \begin{tabular}{|c|c|c|c|}
      \hline
      \multicolumn{2}{|c|}{inverse number}&1&2\\
      \hline
      \multicolumn{2}{|c|}{Nash equilibrium}&$s_2'>v_1(s_1')>v_3(s_3')$&$b_3'>v_1(s_1')>v_2(s_2')$\\
      \hline
      \multirow{2}{*}{bidder of $v_1$}&$2^{nd}\leftrightarrow3^{rd}$&$\underline{V_{2,3}(v_1;v_3,v_4)<Y_{2,3}}$&$\underline{V_{2,3}(v_1;v_2,v_4)>Y_{2,3}}$\\
      \cline{2-4}
      \multirow{2}{*}{}&$1^{st}\leftrightarrow2^{nd}$&$V_{1,2}(v_1;b_2';v_3,v_4)>Y_{1,2}$&$V_{1,2}(v_1;b_3';v_2,v_4)<Y_{1,2}$\\
      \hline
      \multirow{2}{*}{bidder of $v_2$}&$2^{nd}\leftrightarrow3^{rd}$&$\underline{V_{2,3}(v_2;v_3,v_4)<Y_{2,3}}$&$V_{2,3}(v_2;v_1,v_4)>Y_{2,3}$\\
      \cline{2-4}
      \multirow{2}{*}{}&$1^{st}\leftrightarrow2^{nd}$&$V_{1,2}(v_2;v_1;v_3,v_4)<Y_{1,2}$&$V_{1,2}(v_2;b_3';v_1,v_4)>Y_{1,2}$\\
      \hline
      \multirow{2}{*}{bidder of $v_3$}&$2^{nd}\leftrightarrow3^{rd}$&$V_{2,3}(v_3;v_1,v_4)>Y_{2,3}$&$V_{2,3}(v_3;v_2,v_4)<Y_{2,3}$\\
      \cline{2-4}
      \multirow{2}{*}{}&$1^{st}\leftrightarrow2^{nd}$&$V_{1,2}(v_3;b_2';v_1,v_4)>Y_{1,2}$&$V_{1,2}(v_3;v_1;v_2,v_4)<Y_{1,2}$\\
      \hline
    \end{tabular}}
\end{center}
\end{table}

\begin{table}[ht!]
\begin{center}
    \caption{Condition of Nash equilibrium: Part 3}
    \label{tab:table3}
    \resizebox{0.8\linewidth}{!}{ 
    \begin{tabular}{|c|c|c|c|}
      \hline
      \multicolumn{2}{|c|}{inverse number}&2&3\\
      \hline
      \multicolumn{2}{|c|}{Nash equilibrium}&$s_2'>s_3'>v_1(s_1')$&$s_3'>s_2'>v_1(s_1')$\\
      \hline
      \multirow{2}{*}{bidder of $v_1$}&$2^{nd}\leftrightarrow3^{rd}$&$V_{2,3}(v_1;b_3',v_4)>Y_{2,3}$&$V_{2,3}(v_1;b_2',v_4)>Y_{2,3}$\\
      \cline{2-4}
      \multirow{2}{*}{}&$1^{st}\leftrightarrow2^{nd}$&$V_{1,2}(v_1;b_2';b_3',v_4)>Y_{1,2}$&$V_{1,2}(v_1;b_3';b_2',v_4)>Y_{1,2}$\\
      \hline
      \multirow{2}{*}{bidder of $v_2$}&$2^{nd}\leftrightarrow3^{rd}$&$V_{2,3}(v_2;v_1,v_4)<Y_{2,3}$&$V_{2,3}(v_2;v_1,v_4)<Y_{2,3}$\\
      \cline{2-4}
      \multirow{2}{*}{}&$1^{st}\leftrightarrow2^{nd}$&$V_{1,2}(v_2;b_3';v_1,v_4)<Y_{1,2}$&$V_{1,2}(v_2;b_3';v_1,v_4)>Y_{1,2}$\\
      \hline
      \multirow{2}{*}{bidder of $v_3$}&$2^{nd}\leftrightarrow3^{rd}$&$V_{2,3}(v_3;v_1,v_4)<Y_{2,3}$&$V_{2,3}(v_3;v_1,v_4)<Y_{2,3}$\\
      \cline{2-4}
      \multirow{2}{*}{}&$1^{st}\leftrightarrow2^{nd}$&$V_{1,2}(v_3;b_2';v_1,v_4)>Y_{1,2}$&$V_{1,2}(v_3;b_2';v_1,v_4)<Y_{1,2}$\\
      \hline
    \end{tabular}}
\end{center}
\end{table}

The conditions drawn with \underline{ underlines} mean that they are redundant for the Nash equilibrium.
In other words, The conditions drawn with wave lines must be satisfied.
For example, for $V_{2,3}(a;m,n)<0$ when $a>m,a>n$, $V_{2,3}<Y_{2,3}$ must be true because of the non-negativity of $Y_{2,3}$.

\subsubsection{Analysis.}For $\beta=2$, we can get $x_1+x_2=1,x_1\ge x_2$.
So we know $x_1=1-x_2,x_2\in(0,0.5)$, and\\
\centerline{$Y_{2,3}=\displaystyle{\frac{(Y_2-Y_3)x_2}{(Y_1-2Y_4)x_2+Y_4}}$}
(\rmnum{1})If $Y_1<2Y_4$,we can get $x_2\neq\displaystyle{\frac{Y_4}{2Y_4-Y_1}}$ and $\displaystyle{\frac{Y_4}{2Y_4-Y_1}>0.5=x_{2max}}$. 
So we know that $Y_{2,3}$ monotone increases when $x_2\in(0,0.5)$, $Y_{2,3}\le Y_{2,3}(x_2=0.5)=\displaystyle{\frac{Y_2-Y_3}{Y_1}<1}$.

(\rmnum{2})If $Y_1>2Y_4$, $Y_{2,3}$ monotone increases when $x_2\in(0,0.5)$, $Y_{2,3}\le Y_{2,3}(x_2=0.5)=\displaystyle{\frac{Y_2-Y_3}{Y_1}<1}$.

\begin{figure}
    \centering
{
	\begin{minipage}{ 0.3 \linewidth}
	\centering          
	\includegraphics[width=\linewidth]{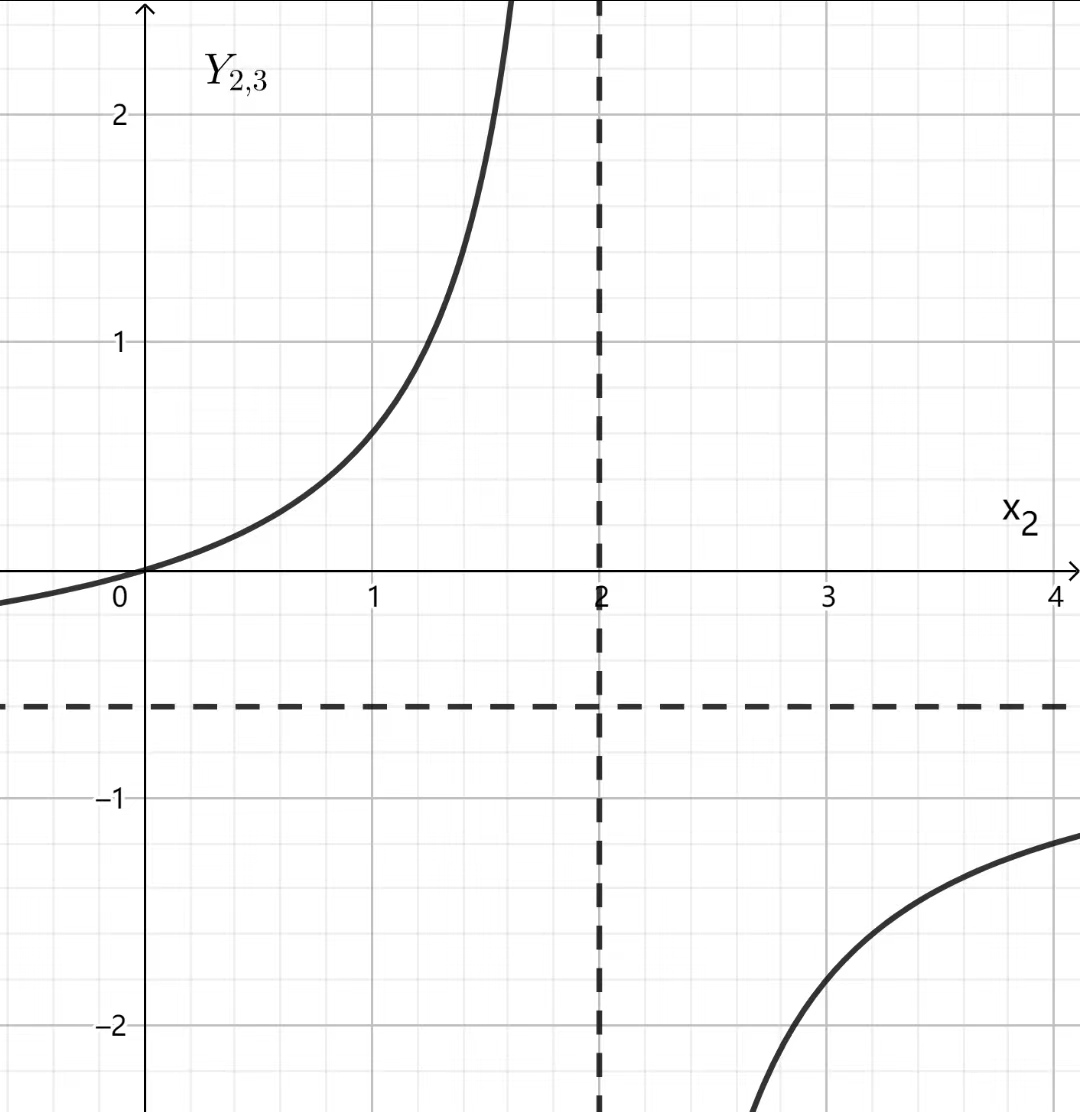}
        \caption{$Y_1<2Y_4$}
	\end{minipage}
}
{
	\begin{minipage}{ 0.3 \linewidth}
	\centering      
	\includegraphics[width=\linewidth]{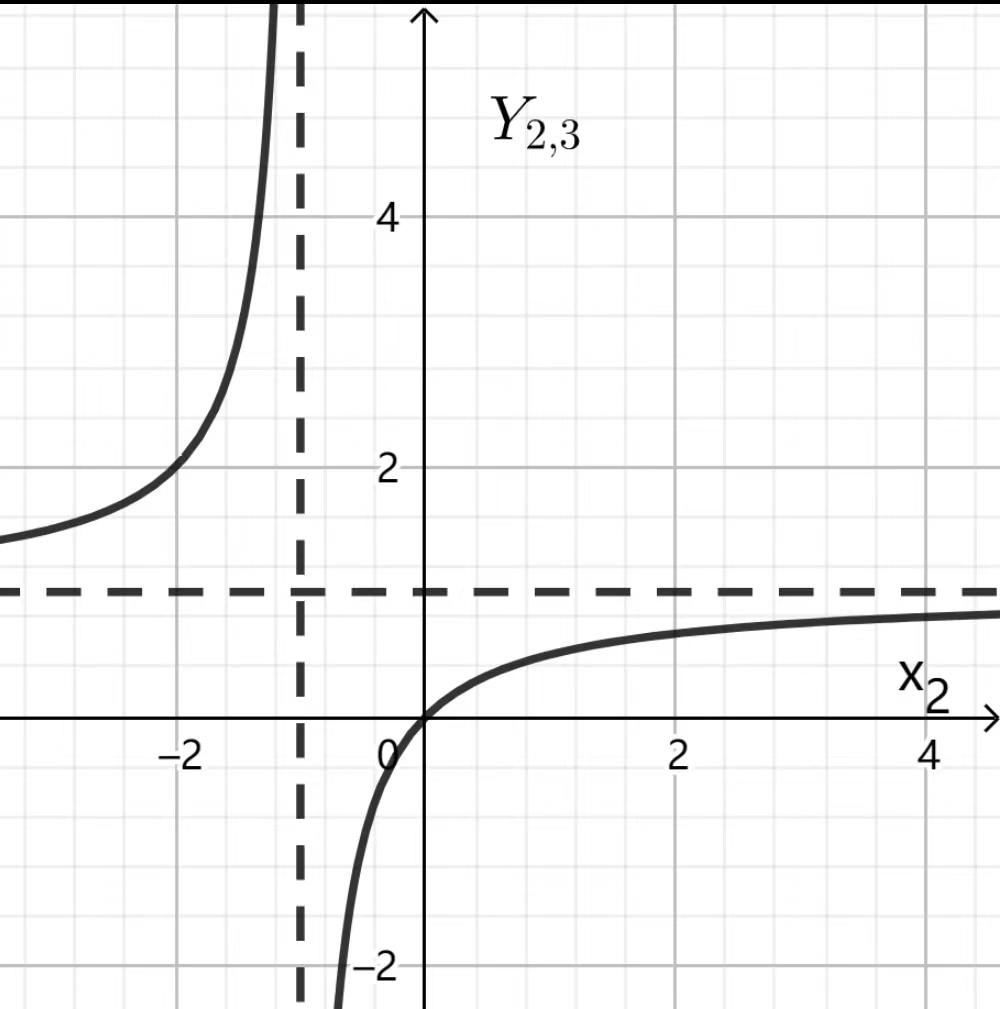}
        \caption{$Y_1>2Y_4$}
	\end{minipage}
}

    \label{fig:y_23}
\end{figure}

We know $Y_{1,2}=\displaystyle{\frac{Y_3-Y_4}{Y_1-Y_2}\in(0,\infty)}$, we must be able to find suitable $Y_{1,2}$ satisfying the conditions of 6 kinds of Nash equilibrium shown in Table~\ref{tab:table1},\ref{tab:table2},\ref{tab:table3}.
For specific $\mathbf{v}=(v_1,v_2,v_3,v_4,v_5)$, we may not find suitable $Y_{2,3}$ satisfying all the conditions of Nash equilibrium because of $Y_{2,3}<1$.
In other words, not all Nash equilibrium listed in Table~\ref{tab:table1},\ref{tab:table2},\ref{tab:table3}.
Additionally, by the partial derivative of $V_{2,3}, V_{1,2}$, we can easily conclude that none of the group of 6 conditions of one specific Nash equilibrium conflicts.
For example, $V_{1,2}(v_1;b_3';b_2',v_4)$ may be bigger than $V_{1,2}(v_3;b_2';v_1,v_4)$ and $V_{1,2}(v_2;b_3';v_1,v_4)$ may be bigger than $V_{1,2}(v_3;b_2';v_1,v_4)$ because of the partial derivative of $V_{1,2}$, so there is no confliction between
$V_{1,2}(v_1;b_3';b_2',v_4)>Y_{1,2},V_{1,2}(v_3;b_2';v_1,v_4)>Y_{1,2},V_{1,2}(v_3;b_2';v_1,v_4)<Y_{1,2}$.\\

\section{Remarks and Future Work}

Our analysis for the DSIC condition of the two-round auction in the Theorem \ref{lemma:two round dsic} still holds if there are pre-defined reserve prices for both rounds.
Meanwhile, one may extend the analysis to multi-round auctions by repeatedly treating the last two rounds as a separate mechanism.
On the other hand, the key observation from our Nash equilibria analysis is that monotone allocation brings greater expected utility for risky bidders while their negative utility can be balanced by the incremental probability to be selected for the second round.
As a result, the risky bidding can increase the revenue of the auctioneer in Nash equilibria of multi-round auctions, although not all of them exist.
As a result, the risky bidding can increase the revenue of the auctioneer in Nash equilibria of multi-round auctions, although not all of them exist.
We conjecture that the greater number of bidders with positive probabilities of being  selected to the second round, the better revenue performance an auction mechanism has, although the more difficult characterizing all Nash equilibria will be.



\bibliographystyle{plain}

\end{document}